\documentclass{article}

\usepackage{arxiv}

\usepackage[utf8]{inputenc} 
\usepackage[T1]{fontenc}    
\usepackage{url}            
\usepackage{booktabs}       
\usepackage{amsfonts}       
\usepackage{nicefrac}       
\usepackage{microtype}      
\usepackage{lipsum}

\usepackage{relsize,balance,lipsum,bbm,enumerate,times,comment,color,graphicx,setspace,mathdots,mathrsfs,amssymb,latexsym,amsfonts,amsmath,cite,stmaryrd,caption,pgf,accents,mathtools,tabu,enumitem,hhline,array,epstopdf,nicefrac,amsthm,microtype,algorithmic,array,float,bm,url}

\usepackage{graphicx}

\usepackage{mathtools} 

\usepackage[utf8]{inputenc}
\usepackage[english]{babel}

\usepackage{tikz}
\usetikzlibrary{positioning}
\definecolor{mygreen}{RGB}{153,255,153}
\definecolor{myorange}{RGB}{255,178,102}
\definecolor{myred}{RGB}{255,153,153}
\definecolor{myblue}{RGB}{153,204,255}

\title{On Characterizations of Potential and Ordinal Potential Games}

\author{
    Sina~Arefizadeh*\\
    Dept. of Electrical and Computer Engineering\\
    Arizona State University\\
    Tempe, Arizona, USA\\
    \texttt{sarefiza@asu.edu}
\And
    Angelia~Nedi\'c\\
    Dept. of Electrical and Computer Engineering\\
    Arizona State University\\
    Tempe, Arizona, USA\\
    \texttt{Angelia.Nedich@asu.edu}
\And
    Gautam~Dasarathy\\
    Dept. of Electrical and Computer Engineering\\
    Arizona State University\\
    Tempe, Arizona, USA\\
    \texttt{gautamd@asu.edu}
}

\theoremstyle{definition}
\newtheorem{lemma}{Lemma}
\newtheorem{thm}{Theorem}
\newtheorem{example}{Example}

\newtheorem{assum}{Assumption}
\newtheorem{definition}{Definition}
\newtheorem{rem}{Remark}

\begin{document}
\thanks{Corresponding author}
\thanks{ This work has been supported in part by the National Science Foundation
under Grants CCF 2106336, CNS 2147641, and CCF 2048223.}
\maketitle
\begin{abstract}
This paper investigates some necessary and sufficient conditions for a game to be a potential game.
At first, we extend the classical results of~\cite{slade1994does} and \cite{monderer1996potential} from games with one-dimensional action spaces to games with multi-dimensional action spaces, which require differentiable cost functions.
Then, we provide a necessary and sufficient conditions for a game to have a potential function by investigating the structure of a potential function in terms of the players' cost differences, as opposed to differentials. This condition provides a systematic way for construction of a potential function, which is applied to network congestion games, as an example. Finally, we provide some sufficient conditions for a game to be ordinal potential and generalized ordinal potential.
\end{abstract}

\section{Introduction}
In game theory, the concept of Nash equilibrium has a significant importance~\cite{maschler2020game}. But, not all games have  a Nash equilibrium. For example, the class of harmonic games does not admit any Nash equilibrium in general and some potential components must exist to admit a Nash equilibrium \cite{candogan2011flows}. Confirming the existence of a Nash equilibrium marks the initial stage in a series of studies. Therefore, significant efforts have been made to understand when a Nash equilibrium is guaranteed and achievable~\cite{C1,C20,Arefi0}. 
Potential games are a major class of games in which a Nash equilibrium is guaranteed and achievable under some mild assumptions. In fact, by examining the potential function with compact lower-level sets, it is straightforward to demonstrate the existence of a Nash equilibrium in games with a set of continuous cost functions.
In many applications within the realm of non-cooperative games, it is a non trivial task to verify if the game is potential. Typically, in this class of games, each player's cost function can be influenced by the decisions made by other players. However, in a wide range of applications, these cost functions depend on a certain function of the aggregate decision variables of all players. For example, in the Cournot game, unlike other players' individual supply, each player's cost function depends on the total supply of their opponents. In recent years, this category of games, which is called aggregative games has garnered interest from various fields, including electrical engineering, economics, and transportation science. References \cite{C4,C7,C8,C9,C10,C11,DarenAcemoghlou,Adler} exemplify such studies. 
Not all aggregative games are necessarily potential games. 

The study of characterizing potential games has long been of interest due to the favorable equilibrium property exhibited in this class of games. However, prior to this study, there was no focused effort on understanding the behavior of potential games or generally characterizing them for multidimensional action spaces, regardless of cost continuity and differentiability. This motivates us to investigate various classes of potential games, such as (exact) potential, ordinal potential, and generalized ordinal potential games, and derive some novel characterizations for them. This enables us to describe other classes of games, such as aggregative games, and determine when they are potential games.

In this paper, we characterize different classes of potential games and significantly improve upon the current characterizations existing in the literature. We also propose a systematic framework for constructing the potential function for  a game with not necessarily continuous cost functions in multidimensional action space. While there exists a robust mathematical foundation for analyzing exact one-differential forms in potential games, which describe these games by characterizing them in terms of cost functions, there is still an open question regarding obtaining a generic form for the potential function that holds true regardless of the continuity of the cost functions. Previous work, such as that by Hwang et al. \cite{Hwang}, has introduced integral tests to determine whether a game is potential, but these tests require calculating integrals over the action space, which can be computationally expensive compared to alternative tests that do not involve integration or differentiation.
Furthermore, the study by Ui et al. \cite{Ui} has investigated the relationship between the Shapley value and the potential function, providing a characterization and systematic methods to construct the potential function in potential games using interaction potential functions. However, there are two main challenges with this approach. Firstly, finding transferable utilities that satisfy the conditions of Theorem~3 in Ui et al. \cite{Ui} to derive the interaction potential functions can be challenging. Secondly, even if these transferable utilities are identified, computing the interaction potential functions becomes increasingly computationally expensive with a larger number of players.
Moreover, the study of ordinal potential games may pose a greater level of difficulty, primarily due to the absence of well-established mathematical structures for analyzing these classes of games. To the best of our knowledge, only a few studies have attempted to address the challenge of characterizing smooth ordinal potential games \cite{ewerhart2020ordinal}. The aforementioned paper and its references have derived some necessary conditions by examining the existence of Nash equilibrium within this class of games.
Consequently, the characterization of aggregative potential games has not received as much attention. While the concept of Best Reply Potential Games introduced in \cite{C22} is acknowledged in the publication \cite{C4} to characterize aggregative best reply potential games, this definition is subject to a debate. Indeed, since the best response is dynamic and dynamics often exist independently of games, creating a new class of games based on the best response correspondence may be misleading. Such a game may not necessarily be expressible in terms of cost functions. Furthermore, while the existence of a Best Reply Potential may be interesting for ensuring the convergence of specific dynamics, such as the best response or better response to a Nash equilibrium, it may not be practical in real-world scenarios where understanding the form of this function is crucial for determining the game's Nash equilibrium. Similarly, it may not be useful for providing convergence analysis for dynamics other than the better or best response.

This paper has the following main contributions:
\begin{itemize}
    \item Firstly, we provide a comprehensive characterization of potential games without the need for integration or differentiation, unlike the existing approaches. We also simplify these characterizations for potential aggregative games. To achieve this, we initially address the challenge of describing potential games in multidimensional cases, aiming to understand the behavior of the potential function in such scenarios. It is important to note that the conditions derived for potential games in terms of pairs of cost functions, as presented in \cite{slade1994does} and \cite{monderer1996potential}, are applicable only to one-dimensional cases, where one differential form corresponds to scalar variables. 
    \item Secondly,  based on the characterization criteria for a potential game, we develop a systematic method to construct a potential function. Furthermore, we apply this systematic construction of potential functions to network congestion game, extending the related results for this game class.
    \item Thirdly, we study ordinal potential games and derive sufficient conditions for a game to be considered ordinal potential.
    \item Finally, we establish some sufficient conditions, in terms of cost functions, for the class of strong and strictly convex games to be generalized ordinal potential.
\end{itemize}

Preliminary versions of Theorem~\ref{nec-suff-cond-2-agg-pot}, Theorem~\ref{theorem 7}, and Theorem~\ref{theorem 8} (specifically, item (a) of each)  have been published in~\cite{Arefi}. In this current paper, these theorems are extended (adding item (b) for each in this study) and improved. In addition, we have new results on the use of these theorems for a systematic construction of potential functions, which is demonstrated in a network congestion game. We also have new results (Theorem~\ref{generic-suff-theo}, Theorem~\ref{final-the}, Theorem~\ref{final-the-one}, and Theorem~\ref{final-the-two}) for ordinal potential and generalized ordinal potential games, which have not been published elsewhere. 

This paper is organized as follows. Section \ref{Se_Nom} introduces the notation and terminologies.  Section \ref{Prelim} provides preliminary concepts regarding potential games in multi-dimensional action spaces, which are used in Section \ref{Main-res} to provide a characterization of potential games and, also, apply the results developed to aggregative games and network congestion games. In Section~\ref{Main-res-2}, sufficient conditions for the characterization of ordinal and generalized ordinal potential games are developed relying on strongly and strictly convex cost functions. Finally, Section \ref{sec:conclusion} concludes the paper.

\section{Notation and Terminology} \label{Se_Nom}
In this section, we provide some definitions and terminologies for the games that will be used in the rest of the manuscript. A game consists of $N$ players represented by the set $\mathcal{N}:=\{1,\dots,N\}$. Each player $i\in\mathcal{N}$ selects an action $x_i$ from a strategy set $K_i \in \mathbb{R}^{n_i}$ to minimize its cost function $f_i:K_i \times K_{-i} \to \mathbb{R}$ where $K_{-i}:= \prod_{j \neq i} K_j$. 
Specifically, each player $i\in\mathcal{N}$ wants to solve the following problem
\begin{eqnarray}\label{min-game}
    &&\min f_i(x_i,x_{-i}), \nonumber\\ 
    && s.t.\ \ \  x_i \in K_i.
\end{eqnarray}
We will use $x$ to denote the vector of joint actions $x_i$, $i\in \mathcal{N}$, i.e., $x=(x_1,\ldots,x_N)$. Given a player set $S\subset\mathcal{N}$, we write $x_{-S}$ to denote the set of decisions of players that are not in the set $S$, i.e., $x_{-S}=(x_i,i\notin S)$. In particular,
for each player $i$, we write $x_{-i}$ to denote the vector consisting of all players' actions except for player $i$, i.e., $x_{-i}=(x_1,\ldots,x_{i-1},x_{i+1},\ldots,x_N)$.
Similarly, we define the set $K_{-S}$ for $S\subset\mathcal{N}$, i.e.,
$K_{-S}=\prod_{j \notin S} K_j$.
A game $\Gamma$ is a tuple of the player set $\mathcal{N}$, the joint action set $K:=\prod_{i \in \mathcal{N}} K_i$, and the collection of cost functions $f_i(\cdot)$, $i \in \mathcal{N}$.
We will consider games with special structures such as potential, ordinal potential, etc., which are defined as follows.

\begin{definition}[Potential Game] \label{def_potential}
A game $\Gamma$ is a potential game if there exists a function $\phi: K \rightarrow \mathbb{R}$ such that the following relation holds for all players $i \in \mathcal{N}$,
and for all $x_i' \in K_i$, $x_i\in K_i$, and $x_{-i} \in K_{-i}$,
\begin{equation}\label{eq_def_potential}
    f_i(x'_i,x_{-i})-f_i(x_i,x_{-i})=\phi(x'_i,x_{-i})-\phi(x_i,x_{-i}).
\end{equation}
\end{definition}

An ordinal potential game is defined as follows.
\begin{definition}[Ordinal Potential Game] \label{def_Best_potential}
A game $\Gamma$ is an ordinal potential game if there exists a function $\phi: K \rightarrow \mathbb{R}$ such that the following relation holds for all players $i \in \mathcal{N}$,
and for all $x_i' \in K_i$, $x_i\in K_i$, and $x_{-i} \in K_{-i}$,
\begin{equation}\label{eq_def-ord_potential}
    f_i(x'_i,x_{-i})-f_i(x_i,x_{-i})<0 \Leftrightarrow \phi(x'_i,x_{-i})-\phi(x_i,x_{-i})<0.
\end{equation}
\end{definition}

A more relaxed version of an ordinal potential game is a generalized ordinal potential game, 
defined as follows.
\begin{definition}[Generalized Ordinal Potential Game] \label{def_Gen_ord_potential}
A game $\Gamma$ is a generalized ordinal potential game if there exists a function $\phi: K \rightarrow \mathbb{R}$ such that the following relation holds for all players $i \in \mathcal{N}$,
and 
for all $x_i' \in K_i$, $x_i\in K_i$, and $x_{-i} \in K_{-i}$,
\begin{equation}\label{eq_def-ord_potential-2}
    f_i(x'_i,x_{-i})-f_i(x_i,x_{-i})<0 \Rightarrow \phi(x'_i,x_{-i})-\phi(x_i,x_{-i})<0.
\end{equation}
\end{definition}
The function $\phi(\cdot)$ in Definitions~\ref{def_potential}--\ref{def_Gen_ord_potential} is referred to, respectively, as a potential, ordinal potential, and generalized ordinal potential function of the game $\Gamma$.

Finally, we introduce an aggregative game. In this type of a game, all players' action sets are of the same dimension,
i.e., $K_i\subset \mathbb{R}^n$ for all $i\in\mathcal{N}$.
We define $\bar{K}$ as the Minkowski sum of the sets $K_i$, i.e.,
\begin{equation}\label{k_bar}
    \bar{K}\triangleq \sum_{i=1}^{N}K_i = \left\{ \sum_{i=1}^Nx_i \mid x_i\in K_i \mbox{ for all } i\in \mathcal{N} \right\}.
\end{equation}
We let $\bar{x}$ be the aggregate of players decisions $x_i$, i.e., 
\begin{equation}\label{x_bar}
    \bar{x}\triangleq \sum_{j=1}^N x_j=x_i+\bar{x}_{-i},\qquad \bar{x}\in \bar{K},
\end{equation}
where $\bar{x}_{-i}$ denotes the aggregate of all players' decisions except for player $i$, i.e., for all $i\in\mathcal{N}$,
$$\bar{x}_{-i}=\sum_{j=1,j\ne i}^N x_j.$$
Each player $i$ is confronted with the following optimization problem:
\begin{eqnarray}\label{game}
    &&\min f_i(x_i,x_{-i})\triangleq \tilde{f}_i\big(x_i,g_i(\bar{x})\big), \nonumber\\ 
    && s.t.\ \ \  x_i \in K_i,
\end{eqnarray}
where $g_i:\bar{K}\rightarrow \mathbb{R}^{m_i}$, with $m_i\ge 1$, is some mapping. Now, we are in position to specify an aggregative game.
\begin{definition}[Aggregative Game \cite{C4}] \label{aggregative_games}
A game $\Gamma=(\mathcal{N},\{\tilde{f}_i,K_i\}_{i\in\mathcal{N}})$, with the strategy sets $K_i\subset\mathbb{R}^n$ and the cost functions $\tilde{f}_i$ as in~\eqref{game}, is an aggregative game.
\end{definition}
Basically, in an aggregative game, a cost function of each player depends on its own strategy $x_i$ and the aggregate $\bar x$ of the strategies of all players.

\section{Preliminaries}\label{Prelim}
\def\la{\langle}
\def\ra{\rangle}
The characterization of potential games in~\cite{slade1994does} and \cite{monderer1996potential} is based on one-differential forms \cite{Rudin-Analysis}, so it is limited to one-dimensional action spaces (i.e., $K_i\subset\mathbb{R}$ for all $i$). In this section, we focus on the class of potential games for the multidimensional case, $K_i\subset\mathbb{R}^{n_i}$. The paper ~\cite{deb2008interdependent} extends~\cite{monderer1996potential} results for games with smooth cost functions from one dimensional to multiple dimensional action space. However, in the reminder of this section we will cover similar results to \cite{deb2008interdependent} through a different approach.

 
 One-form $w_\mathbf{a}(v)$, at an arbitrary point $\mathbf{a}\in\mathbb{R}^n$, is  a linear functional on the space of tangent vectors $v$ at the point  $\mathbf{a}=(a_1,a_2,\ldots,a_n)$.
 Considering the tangent vector $d\mathbf{a}$ at point $\mathbf{a}$, there is a unique function $F:\mathbb{R}^n\rightarrow\mathbb{R}^n$ such that
\begin{equation}\label{one-form}
    w_\mathbf{a}(d\mathbf{a})=\la F(\mathbf{a}),d\mathbf{a}\ra,
\end{equation}
where $\la\cdot,\cdot\ra$ denotes the inner product in $\mathbb{R}^n$.

The following theorem gives a condition for a one-differential form to be exact over a convex set in $\mathbb{R}^n$.
\begin{thm}\label{exact-form}
Let $u_i(\cdot)$ be a continuously differentiable function  on a convex set $E\subset \mathbb{R}^n $,
for all $i=1,2,\ldots,n$, for some $n\ge1$.
Then, the one-form $\omega_\mathbf{a}(d\mathbf{a})=\sum_{i=1}^{n} u_i(\mathbf{a})da_i$, with $\mathbf{a}\in E$, is exact on the set $E$ if and only if we have for all $i,j \in \{1,2,\ldots,n\}$,
\begin{equation}\label{nec-suf-exact}
    \frac{\partial u_i(\mathbf{a})}{\partial a_j}=\frac{\partial u_j(\mathbf{a})}{\partial a_i}\qquad\hbox{for all }\mathbf{a}\in E.
\end{equation}
\end{thm}
\begin{proof}
The proof follows immediately from Remark 10.35(a) and Theorem 10.39 of ~\cite{Rudin-Analysis}.
\end{proof}
Using Theorem~\ref{exact-form}, a vector field $F:\mathbb{R}^I\rightarrow\mathbb{R}^I$ we can uniquely define one-form as given in \eqref{one-form}. For an arbitrary game,
one can introduce a one-form for each by viewing the concatenation of derivatives of the cost function of each player with respect to its own decision variable. 
For example, for a game $\Gamma=(\mathcal{N},\{K_i,f_i\}_{i\in\mathcal{N}})$, with one dimensional action space, one can consider the vector field 
$$G(x)=\left(\frac{\partial f_1}{\partial x_1},\frac{\partial f_2}{\partial x_2},\ldots,\frac{\partial f_N}{\partial x_N}\right),$$
and assign one-form $\langle G(x), dx \rangle$ to this game $\Gamma$. Based on the definition of the exact one-form, if the one-form corresponding to this game is exact, the game is potential since we can write
$$\langle G(x), dx \rangle= \sum_{i=1}^N \frac{\partial \phi(x)}{\partial x_i}dx_i$$ for some function $\phi:K\rightarrow\mathbb{R}$, with $K\subset\mathbb{R}^N$.

Let $\bar{n}=\sum_{i=1}^{N}n_i$, where $n_i$ is the dimension of the action variable for player $i$, i.e., $K_i\subset \mathbb{R}^{n_i}$. We define one-differential forms on $K$.
To do so, for each player $i\in \mathcal{N}$, we consider  a curve $x_i(\cdot):\mathbb{R}\to K$, which is continuously differentiable and given by 
\begin{align*}
    x_i(t)= &\big(x_{11}(t_0),\ldots,x_{1n_1}(t_0),\ldots,x_{i1}(t),\ldots,\nonumber \\
    &x_{in_i}(t),\ldots,x_{N1}(t_0),\ldots,x_{Nn_N}(t_0)\big)\in \mathbb{R}^{\bar{n}},
\end{align*}
for some $t_0\in \mathbb{R}$. Note that, for all $i,j\in \mathcal{N}$ we have $x_i(t_0)=x_j(t_0)$. 
We next describe potential games in a parametric sense, as follows:
A game $\Gamma$ is potential if there is a scalar function $\phi$ such that
\begin{equation}\label{eq-MDAS-pot}
    \frac{df_i\big(x_i(t)\big)}{dt}=\frac{d\phi\big(x_i(t)\big)}{dt}
\end{equation}
for arbitrary curves $x_i(\cdot)$ in the set $K$ and for all $i\in \mathcal{N}$. 
The preceding relation describes a parametric version of a potential game, and 
we refer to $\phi$ as a potential function of the game. 
\begin{lemma}\label{lem1}
Consider the game $\Gamma=(\mathcal{N},\{f_i,K_i\}_{i\in\mathcal{N}})$ where each function $f_i$ is continuously differentiable over an open set containing the joint action set $K$. Then, the game is a potential game if and only if there exists a function $\phi:K\to\mathbb{R}$
such that $\frac{\partial f_i(x)}{\partial x_{im}}=\frac{\partial \phi(x)}{\partial x_{im}}$
for all $x\in K$, and for all $i\in \mathcal{N}$ and $m \in \{1,2,\ldots,n_i\}$.
\end{lemma}
\begin{proof}
If $\Gamma$ is a potential game, defining $x_{il}(t)=x_{il}(t_0)$ for all $l\in \{1,2,\ldots,n_i\}\setminus\{m\}$, according to \eqref{eq-MDAS-pot}, we have $\frac{\partial f_i}{\partial x_{im}}dx_{im}=\frac{\partial \phi}{\partial x_{im}}dx_{im}$. Since $x_{im}$ is selected arbitrarily, we can conclude that $\frac{\partial f_i}{\partial x_{im}}=\frac{\partial \phi}{\partial x_{im}}$. 
For the converse statement, since 
$\frac{\partial f_i(x)}{\partial x_{im}}=\frac{\partial \phi(x)}{\partial x_{im}}$ for all $x\in K$ and 
for all $i\in \mathcal{N}$ and $m \in \{1,2,\ldots,n_i\}$, by the differentiation rule it follows that
\begin{align*}
\frac{df_i\big(x_i(t)\big)}{dt} 
&=\sum_{m=1}^{n_i}\frac{\partial f_i}{\partial x_{im}}\cdot\frac{dx_{im}}{dt} \cr
&=\sum_{m=1}^{n_i}\frac{\partial \phi}{\partial x_{im}}\cdot\frac{dx_{im}}{dt}\cr 
&=\frac{d\phi\big(x_i(t)\big)}{dt}.
\end{align*}
\end{proof}

By Theorem~\ref{exact-form} and Lemma~\ref{lem1}, we have the following result.
\begin{thm}\label{Th-MDASPG}
Consider the game $\Gamma=(\mathcal{N},\{f_i,K_i\}_{i\in\mathcal{N}})$.
\begin{itemize}
    \item[(a)] Suppose that each cost function $f_i:K\to\mathbb{R}$ is continuously differentiable over an open set containing the set $K\subset\mathbb{R}^{\bar n}$. Then, the game $\Gamma$ is a potential game if and only if there exists function $\phi:\mathbb{R}^{\bar n}\rightarrow \mathbb{R}$ such that 
    or all $i,j\in\mathcal{N}$,
    \[\nabla_{x_i} f_i(x)=\nabla_{x_i}\phi(x)\qquad\hbox{for all }x\in K,\]
    where $\nabla_{x_i}$ is the partial gradient with respect to $x_i$.
    \item[(b)] Suppose that each cost function $f_i$ is twice continuously differentiable over an open set containing the set $K$ and the set $K$ is convex. Then, the game $\Gamma$ is a potential game if and only if 
    for all $i,j\in\mathcal{N}$,
    \[\nabla^2_{x_i,x_j} f_i(x)=\nabla^2_{x_j,x_i}f_j(x)\qquad\hbox{for all }x\in K,\]
    where $\nabla^2_{x_i,x_j}$ is the  partial Hessian with respect to $x_i,x_j$.
\end{itemize}
\end{thm}


\begin{proof}
Part (a) follows from Lemma~\ref{lem1}.
For part (b), consider a joint strategy 
$x=(x_1,\ldots,x_N)$, which we relabel
for convenience, as follows:
\[x=(y_1,y_2,\ldots,y_{\bar n})=\mathbf{y}.\]
Now, we consider the differential one-form
\[
w_\mathbf{y}(d\mathbf{y})=\sum_{s=1}^{\bar{n}} \frac{\partial u_s(\mathbf{y})}{\partial y_{s}}dy_{s}\qquad\hbox{for }\mathbf{y}\in K,\]
where $u_s=f_i$ for $s=n_{i-1}+1,\ldots, n_i$
and for all $i\in \mathcal{N}$. 
According to (\ref{nec-suf-exact}), where $E=K$, the differential one-form $w_\mathbf{y}$ is exact if and only if 
for all $l,d \in \{1,2,\ldots,\bar{n}\}$,
\[\frac{\partial}{\partial y_{l}}\left( \frac{\partial u_d(\mathbf{y})}{\partial y_{d}}\right)=\frac{\partial}{\partial y_{d}} \left(\frac{\partial u_l(\mathbf{y})}{\partial y_{l}}\right)\qquad\hbox{for all }\mathbf{y}\in K.\]
$u_s=f_i$ for $s=n_{i-1}+1,\ldots, n_i$
and for all $i\in \mathcal{N}$. Since $x=\mathbf{y}$, it follows that the game $\Gamma$ is potential if and only if 
$$ \frac{\partial^2 f_i(x)}{\partial x_{ip}\partial x_{jq}}
=\frac{\partial^2 f_j(x)}{\partial x_{jq}\partial x_{ip}}\qquad\hbox{for }x\in K,$$
for all $i,j \in \mathcal{N}$, $p\in \{1,2,\ldots,n_i\}$, and $q \in \{1,2,\ldots,n_j\}$. 
\end{proof}

\section{Characterization of Potential Games}\label{Main-res}
In this section, we provide some necessary and sufficient conditions for a game to have a potential function. We do so by investigating the structure of a potential function in terms of differences, as opposed to differentials. 

\subsection{Necessary Condition}
We start by considering the potential functions from the parametric sense. 
According to \eqref{eq-MDAS-pot}, for every $i\in\mathcal{N}$, by integrating both sides of the equation and using the Stokes theorem, we can write
\begin{equation}
\label{char1-pot}
   \phi\big(x_i(t_0+\epsilon)\big)-\phi\big(x_i(t_0)\big)=f_i\big(x_i(t_0+\epsilon)\big)-f_i\big(x_i(t_0)\big),
\end{equation}
where $\epsilon$ is any value for which the curve $x_i(t)$ stays within the set $K$.
In the non-parametric form, relation~\eqref{char1-pot} can alternatively be stated as follows: for every $z\in K$, every $y_i\in K$, with $z_i+y_i\in K$, and for every $i\in\mathcal{N}$,
\begin{equation}
\label{char1-pot-non-para}
   \phi(z_i+y_i,z_{-i})-\phi(z_i,z_{-i})=f_i(z_i+y_i,z_{-i})-f_i(z_i,z_{-i}).
\end{equation}
The preceding relation for a potential function does not require differentiability or continuity of the players' cost functions.

The following theorem provides a necessary condition for the form of a potential function when a game is potential.

\begin{thm}\label{theorem-3}
Let a game $\Gamma=(\mathcal{N},\{f_i,K_i\}_{i\in\mathcal{N}})$ be a potential game. Then, any potential function $\phi$ of the game satisfies the following relation:
for every $z\in K$ and $z+y\in K$,
\begin{align}\label{char2-pot-non-para}
    \phi(z+y)-\phi(z)&=\sum_{i=1}^N\big(f_i(z_1+y_1,\ldots,z_i+y_i,z_{i+1},\ldots,z_N)\nonumber \\
    &-f_i(z_1+y_1,\ldots,z_{i-1}+y_{i-1},z_{i},\ldots,z_N)\big).
\end{align}
\end{thm}
\begin{proof}
 Consider a path $P:z\rightarrow(z_1+y_1,z_{-1})\rightarrow(z_1+y_1,z_2+y_2,z_{-\{1,2\}})\rightarrow \ldots\rightarrow z+y$, where $z\in K$ and $z+y\in K$ are arbitrary. By using \eqref{char1-pot-non-para} for every two sequential components of this path and summing them over $i=1,\ldots, N$, we obtain \eqref{char2-pot-non-para}. 
\end{proof}

In the sequel, the following definition of an abnormal game will be used.
\begin{definition}[Abnormal Game]\label{ab-game} A game $(\mathcal{N},\{f_i,K_i\}_{i\in\mathcal{N}})$ is an abnormal game if there is an $i\in\mathcal{N}$ such that for every $x_{-i}\in K_{-i}$ and for every $x_{i}\in K_{i}$ we have $f_i(x_i,x_{-i})=C_i(x_{-i})$ for some function $C_i:K_{-i}\rightarrow \mathbb{R}$.
\end{definition}
Thus, in an abnormal game, there is a player whose actions are not affecting its own cost function but it may affect other players' cost functions. 
In this case, there is no incentive for such a player to make any decision with respect to other players' decisions. 
In potential games that are abnormal, the potential function is not sensitive to the decision variables of such players. 

It turns out that an aggregative game that is not abnormal has some interesting properties. To formally state this,
we start by exploring the expression on the right hand-side of relation~\eqref{char2-pot-non-para}. Given  arbitrary $z\in K$ and $z+y\in K$, define the path $P$, as follows:
\begin{equation}\label{eq-path}
z\to(z_1+y_1,z_{-1})\to(z_1+y_1,z_2+y_2,z_{-\{1,2\}})\rightarrow\ldots\rightarrow z+y.
\end{equation}
Next, we denote the right-hand side of \eqref{char2-pot-non-para} by $h_P(z,y)$,
i.e.,
\begin{align}\label{def-hp}
h_P(z,y)&=\sum_{i=1}^N\big(f_i(z_1+y_1,\ldots,z_i+y_i,z_{i+1},\ldots,z_N)\nonumber \\
    &-f_i(z_1+y_1,\ldots,z_{i-1}+y_{i-1},z_{i},\ldots,z_N)\big),
    \end{align}
which corresponds to incremental differences in players' costs along the path $P$.
Note that 
\[h_P(z,0)=0 \qquad\hbox{for all }z\in K.\]

The following theorem shows that, in aggregative game  that is not abnormal,  
the function $h_P(z,y)$ is not equal to zero, when viewed as a function of $y$, given $z$. 

\begin{thm}\label{thm-agg}
Consider an aggregative game $\Gamma=(\mathcal{N},\{\tilde{f}_i,K_i\}_{i\in\mathcal{N}})$ that is not abnormal. Then, the following statements are true: 
\begin{itemize}
    \item [(a)] If $0\in K$, then 
    \[h_P(0,y)\ne 0\qquad\hbox{for all $y\in K$}.\]
    \item[(b)] If each cost function $f_i$ is defined on $\mathbb{R}^{\bar n}$,
    then
    \[h_P(0,y)\ne 0\qquad\hbox{for all $y\in \mathbb{R}^{\bar n}$}.\]
\end{itemize}
\end{thm}
\begin{proof}
(a) 
We prove the statement via contradiction.
Assume on contrary that 
$h_P(0,y)= 0$ for all $y\in K$.
Consider $y=(0,\ldots,0,u,v)$ 
with $u\in K_{N-1}$ and $v\in K_N$. Thus, we have
\begin{align}\label{proof-prop-2}
    h_P(0,z)&=\tilde{f}_{N-1}\big(u,g_{N-1}(u)\big)-\tilde{f}_{N-1}\big(0,g_{N-1}(0)\big)\nonumber \\
    &+\tilde{f}_{N}\big(v,g_N(u+v)\big)-\tilde{f}_{N}\big(0,g_N(u)\big)=0.
\end{align}
Using $v=0$ in~\eqref{proof-prop-2}, we find that 
\begin{equation}\label{proof-prop-2-1}
   \tilde{f}_{N-1}\big(u,g_{N-1}(u)\big)-\tilde{f}_{N-1}\big(0,g_{N-1}(0)\big)=0.
\end{equation}
Therefore, for all $v\in K_N$,
\begin{equation}\label{proof-prop-2-2}
    \tilde{f}_{N}\big(v,g_{N}(u+v)\big)=\tilde{f}_{N}\big(0,g_{N}(u)\big).
\end{equation}
This implies that the cost value $\tilde{f}_{N}\big(v,g_{N}(u+v)\big)$ 
is constant on the set $K_N$. Therefore,
for every $x_{-N}$ we have 
$f_N(x_N,x_{-N})=f_N(0,x_{-N})$. Hence, the game is abnormal (for player $N$), which is  a contradiction.

(b) The proof 
follows along the same lines as in part (a) by noting that the restrictions to the sets $K$, $K_{N-1}$, and $K_N$ are replaced respectively by $\mathbb{R}^{\bar n}$, $\mathbb{R}^{n_{N-1}}$, and $\mathbb{R}^{n_N}.$ This will lead to a conclusion that \eqref{proof-prop-2-2} holds for all $v\in \mathbb{R}^{n_N}$, which in turn implies that the relation is also valid for $v\in K_N$. The rest follows as in part (a).
\end{proof}

\begin{rem}
Excluding aggregative games, in the general case, there might be non-abnormal games that satisfy \eqref{proof-prop-2}, \eqref{proof-prop-2-1}, and \eqref{proof-prop-2-2}. 
For instance, consider a game with $f_i(x)=\prod_{j=1}^N x_j$ for all $i\in \mathcal{N}$, and with a joint action set $K$ that contains the origin.
\end{rem}

\subsection{Necessary and Sufficient Conditions}
In the following theorem, we provide a necessary and sufficient condition for a game to have a  potential.
\begin{thm}\label{nec-suff-cond-agg-pot}
 The game $\Gamma=(\mathcal{N},\{f_i,K_i\}_{i\in\mathcal{N}})$ is a potential game if and only if there exists a function $\phi$ such that 
\begin{align}\label{char2-agg-pot-non-para}
    \phi(z+y) &-\phi(z)=\sum_{i=1}^N(f_i(z_1+y_1,\ldots,z_i+y_i,z_{i+1},\ldots,z_N)\cr
    &-f_i(z_1+y_1,\ldots,z_{i-1}+y_{i-1},z_{i},\ldots,z_N)),
\end{align}
for all $z\in K$ and $y$ with $z+y\in K$. 
Moreover, the function $\phi$ is a potential function of the game.
\end{thm}
\begin{proof}
If the game $\Gamma$ is a potential game, then 
by Theorem~\ref{theorem-3}
relation \eqref{char2-agg-pot-non-para} holds. For the converse statement, the existence of a function $\phi$ such that \eqref{char2-agg-pot-non-para} holds
implies that for any $z\in K$
any $y=(0,\ldots,y_i,\ldots,0)$, 
with $z_i+y_i\in K_i$,  and for every $i\in \mathcal{N}$, we have that 
$\phi(z_i+y_i,z_{-i})-\phi(z_i,z_{-i})=f_i\big(z_i+y_i,z_{-i}\big)-f_i\big(z_i,z_{-i}\big)$. Thus, the game is a potential game and $\phi$ a potential function. 
\end{proof}

In the sequel, the following notion will be used.
\begin{definition}\label{def-genpot}
Given a collection of functions $\{f_i, i\in {\cal N}\}$, where each $f_i:\mathbb{R}^{\bar n}\to\mathbb{R}$,
we say that $\phi:\mathbb{R}^{\bar n}\to\mathbb{R}$ is a {\it global potential for the collection $\{f_i, i\in {\cal N}\}$}, when the following holds: for all $i\in {\cal N},$
all $x_i,x'_i\in\mathbb{R}^{n_i}$, and all $x_{-i}\in \mathbb{R}^{\bar n- n_i},$
\[f_i(x_i,x_{-i})-f_i(x'_i,x_{-i})
=\phi(x_i,x_{-i})- \phi(x'_i,x_{-i}).\]
\end{definition}
This definition coincides with Definition~\ref{def_potential},
when $K=\mathbb{R}^{n_i}$.
However, we want to consider the notion of global potential as the property of a given collection $\{f_i, i\in {\cal N}\}$ of functions irrespective of a particular game that may be associated with it. That is, 
if a given collection $\{f_i, i\in {\cal N}\}$
has a global potential $\phi$, then for any collection $\{K_i,i\in{\cal N}\}$, with $K_i\subset\mathbb{R}^{n_i}$, 
the game $\Gamma=(\mathcal{N},\{f_i,K_i\}_{i\in\mathcal{N}})$ is a potential game, with a potential function $\phi$.

In the next theorem, we establish a necessary and sufficient conditions for a game to be potential assuming that the strategy sets $K_i$ are symmetric,
where a set $X$ is said to be symmetric if 
for every $x\in X$  we have that $-x\in X$.
\begin{thm}\label{nec-suff-cond-2-agg-pot}
Assume that either one of the following cases holds true:
\begin{enumerate}
    \item [(a)] $0\in K$ and the set
$K_i$ is symmetric for all $i \in \mathcal{N}$.
    \item [(b)] Each cost function $f_i(\cdot)$, for $i\in \mathcal{N}$, is defined on the entire set $\mathbb{R}^{\bar{n}}$.
\end{enumerate}
Then, the game $\Gamma=(\mathcal{N},\{f_i,K_i\}_{i\in\mathcal{N}})$ is a potential game if and only if $h_P(z,y)=h_P(0,z+y)-h_P(0,z)$, and the potential function is $\phi(z)=C-h_P(z,-z)$, where $C$ is some constant.
\end{thm}
\begin{proof}
(a)\ By Theorem~\ref{nec-suff-cond-agg-pot}(a),
the definition of $h_P(z,y)$ in~\eqref{def-hp}, we have that the game is potential if and only if 
\begin{equation}\label{eq-hpmid0}
h_P(z,y)=\phi(z+y)-\phi(z)
\end{equation}
for all $z\in K$ and $y$ with $z+y\in K$.
Then, using the symmetry of the set $K$ and relation
~\eqref{char2-agg-pot-non-para} with $y=-z$,
we have that the game is potential if and only if 
\begin{equation}\label{eq-hpmid}
h_P(z,-z)=\phi(0)-\phi(z)\qquad\hbox{for all $z\in K$},
\end{equation}
and  a potential function is given by
$\phi(z)=\phi(0)-h_P(z,-z)$.

It remains to show that 
relation~\eqref{eq-hpmid} implies that $h_P(z,y)=h_P(0,z+y) - h_P(0,z)$.
To show this, we note that from~\eqref{eq-hpmid}  it follows that
\begin{equation}\label{eq-hpmid2}
h_P(z+y,-(z+y))=\phi(0)-\phi(z+y).
\end{equation}
Thus, using~\eqref{eq-hpmid0} we obtain
\begin{align}\label{eq-hpmid3}
h_P(z,y)&=\phi(x+z)-\phi(z)\cr
&=\phi(x+z)+h_P(z,-z)-\phi(0)\cr
&=\phi(0)-h_P(z+y,-(z+y))+h_P(z,-z)-\phi(0)\cr&
=h_P(z,-z)-h_P\left(z+y,-(z+y)\right),
\end{align}
where the first equality is obtained using~\eqref{eq-hpmid},
while the second one is obtained using~\eqref{eq-hpmid2}.
By letting $z=0$ in the preceding relation, we obtain 
\begin{equation}\label{eq-hpmid4}
h_P(0,y)=h_P(0,0)-h_P(y,-y).
\end{equation}
Letting $y=0$ in~\eqref{eq-hpmid4},  we find that  $h_P(0,0)=0$,
implying that
\[h_P(0,y)=-h_P(y,-y).\]
Thus, $h_P(y,-y)=-h_P(0,y)$ and using this relation in equation~\eqref{eq-hpmid3} yields
\[h_P(z,y)=h_P(0,z+y)-h_P(0,z).\]
(b) The proof of part (b) follows the same line of arguments as in part (a), where $K$ is replaced with $\mathbb{R}^{\bar n}$.
\end{proof}

A question may raise about what relation between potential games and duopoly potential games holds. In the sequel, we investigate this relation and show that a necessary and sufficient condition for an $N$-player game to be potential is that every $2$-player game of a suitably defined subgame is potential game.

Consider a pair $i,j\in \mathcal{N}$ of distinct players, and 
without loss of generality, assume that 
$i<j$. We define $h_{ij}(z_i,z_j,y_i,y_j;z_{-\{i,j\}})$ as follows:
\begin{align}\label{char-agg-pot-def-1}
h_{ij}(z_i,z_j,y_i,y_j;z_{-\{i,j\}})
    =&f_{i}\big(z_i+y_i,z_j;z_{-\{i,j\}}\big)\nonumber \\
    &-f_{i}\big(z_i,z_j;z_{-\{i,j\}}\big)\nonumber \\
    &+f_{j}\big(z_j+y_j,z_i+y_i;z_{-\{i,j\}}\big)\nonumber \\
    &-f_{j}\big(z_j,z_i+y_i;z_{-\{i,j\}}\big)
\end{align}
for every $z_i\in K_i$, $z_j\in K_j$, $z_{-\{i,j\}}\in K_{-\{i,j\}}$,
and $y_i$ and $y_j$ with $z_i+y_i\in K_i$ and 
$z_j+y_j\in K_j$.

Aside from the paths from $z$ to $z+y$ for two-player strategies, we will also use longer paths. In particular,
a finite path in the strategy space $K$ is a sequence $(x^1,x^2,\ldots,x^m)$ of elements $x^k\in K$ such that the strategies $x^k$ and $x^{k+1}$ differ only in an action of a single player $i_k$, for all $k=1,2,\ldots,m-1$. In other words, starting from a  strategy profile $x^1\in K$, the next strategy $x^2$ is obtained by a single player $i_1$  deviating from its decision $x_{i_1}^1$, and so on.

For a finite path $\mathcal{Q}=(q^0,\ldots,q^\ell)$, where $q^l \in K$ for $l\in\{0,1,\ldots,\ell\}$, and a collection $f=(f_1,\ldots,f_N)$ of the cost functions $f_i:K\rightarrow \mathbb{R}$, we  consider
the following quantity:
\begin{equation}\label{shap_c}
    I(\mathcal{Q},f)=\sum_{e=0}^{\ell-1}\big(f_{i_{e}}(q^{e+1})-f_{i_{e}}(q^{e})\big),
\end{equation}
where, $i_e$ is the unique deviator at step $e$ (i.e., $q^{e}_{i_e}\neq q^{e+1}_{i_{e}}$). Note that this quantity $I(\mathcal{Q},f)$ has been used in \cite{monderer1996potential} to study potential games.
In particular, in \cite{monderer1996potential}, it
has been proven that a game $\Gamma=(\mathcal{N},\{f_i,K_i\}_{i\in\mathcal{N}})$ is potential if and only if $I(\mathcal{Q},f)=0$ for every finite simple closed path $\mathcal{Q}$ of length $4$; here, 
a path $\mathcal{Q}$ is closed if 
$q^0=q^\ell$, and it is simple when the strategies $q^l$, $l=1,\ldots,\ell$ are distinct. The length of a path is the number of distinct strategies in the path.

Using the special paths of length 2, i.e., the quantities $h_{ij}$ defined in
~\eqref{char-agg-pot-def-1}, we have the following result.

\begin{thm}\label{theorem 7}
Assume that either one of the following cases holds true:
\begin{enumerate}
    \item [(a)] $0\in K$ and the set
$K_i$ is symmetric for all $i \in \mathcal{N}$.
    \item [(b)] The cost function $f_i(\cdot)$ is defined on the entire set $\mathbb{R}^{\bar{n}}$, for all $i\in \mathcal{N}$ .
\end{enumerate}
Then, the game $\Gamma=(\mathcal{N},\{f_i,K_i\}_{i\in\mathcal{N}})$ is potential if and only if for all $i,j\in \mathcal{N}$ and $z_{-\{i,j\}}\in K_{-\{i,j\}}$ we have 
\begin{align}
\label{char-agg-pot-def-0}
h_{ij}(z_{i},z_{j},y_{i},y_{j};z_{-\{i,j\}}) 
    =&h_{ij}(0,0,z_{i}+ y_{i},z_{j}+ y_{j};z_{-\{i,j\}})\nonumber \\
    &-h_{ij}(0,0,z_{i},z_{j};z_{-\{i,j\}}),
    \end{align}
    where $h_{ij}$ is defined in~\eqref{char-agg-pot-def-1}.

\end{thm}
\begin{proof}
(a) 
Let the game $\Gamma=(\mathcal{N},\{f_i,K_i\}_{i\in\mathcal{N}})$ be potential. Then, according to Theorem 2.8 of~\cite{monderer1996potential}, for every closed path $\mathcal{Q}$ of length $4$ we have $I(\mathcal{Q},f)=0$. Let 
\begin{align}\label{eq-simplepath}
    \mathcal{Q}:z&\rightarrow(z_1,\ldots,z_{i}+y_{i},\ldots,z_N) \nonumber \\
    &\rightarrow (z_1,\ldots,z_{i}+y_{i},\ldots,z_{j}+y_{j},\ldots,z_N)\nonumber \\
    &\rightarrow(z_1,\ldots,z_{i},\ldots,z_{j}+y_{j},\ldots,z_N)\cr
    &\rightarrow z,\qquad
\end{align} 
where $z=(z_1,z_2,\ldots,z_N)$. Let  $y=(y_{i},y_{j};0_{-\{i,j\}})$. 
Consider the path $P$ from $z$ to $z+y$ as given in~\eqref{eq-path}. 
By Theorem~\ref{nec-suff-cond-2-agg-pot}, in a potential game, we have that 
\begin{align}\label{theorem-item-3-1}
    h_{ij}\big(z_i,z_j,y_i,y_j;z_{-\{i,j\}}\big)
    &=h_{P}(z,y)\cr
    &=h_{P}(0,z+y)-h_{P}(0,z).\qquad
\end{align}
Since $h_{ij}\big(z_i,z_j,y_i,y_j;z_{-\{i,j\}}\big)
    =h_{P}(z,y)$, for the right hand side of \eqref{theorem-item-3-1} we have that
\begin{align}\label{theorem-item-3-1-2}
h_{P}(0,z+y)-h_{P}(0,z)
= & h_{ij}\big(0,0,z_{i}+y_i,z_{j}+y_j;z_{-\{i,j\}}\big) \nonumber\\
    &-h_{ij}\big(0,0,z_{i},z_{j};z_{-\{i,j\}}\big).
\end{align}
By combining relations~\eqref{theorem-item-3-1} 
and~\eqref{theorem-item-3-1-2},
we obtain the stated relation in~\eqref{char-agg-pot-def-0}.

For the converse statement, let assume \eqref{char-agg-pot-def-0} holds for every $i,j\in \mathcal{N}$ and $z_{-\{i,j\}}\in K_{-\{i,j\}}$. Let 
$\mathcal{Q}$ be an arbitrary simple closed path of length $4$, as given in~\eqref{eq-simplepath}. We expand the $I(\mathcal{Q},f)$ along the path $\mathcal{Q}$ and then collect relevant terms to use functions $h_{ij}(\cdot), $ as follows:
\begin{align*}
    I(\mathcal{Q},f) 
    = &
    f_i(z_i+y_i,z_j,z_{-\{i,j\}})-f_i(z_i,z_j,z_{-\{i,j\}})\nonumber\\
    &+f_j(z_j+y_j,z_i+y_i,z_{-\{i,j\}})\nonumber\\
    &-f_j(z_j,z_i+y_i,z_{-\{i,j\}})\nonumber\\
    &+f_i(z_i,z_j+y_j,z_{-\{i,j\}})\nonumber\\
    &-f_i(z_i+y_i,z_j+y_j,z_{-\{i,j\}})\nonumber\\
    &+f_j(z_j,z_i,z_{-\{i,j\}})-f_j(z_j+y_j,z_i,z_{-\{i,j\}})\nonumber\\
    = &
    h_{ij}(z_{i},z_{j},y_{i},y_{j};z_{-\{i,j\}})\nonumber\\
    &+h_{ij}(z_{i}+y_{i},z_{j}+y_{j},-y_{i},-y_{j};z_{-\{i,j\}})\nonumber\\
    = &
    h_{ij}(0,0,z_{i}+y_{i},z_{j}+y_{j};z_{-\{i,j\}})\nonumber\\
    &-h_{ij}(0,0,z_{i},z_{j};z_{-\{i,j\}})\nonumber\\
    &+h_{ij}(0,0,z_{i},z_{j};z_{-\{i,j\}})\nonumber\\
    &-h_{ij}(0,0,z_{i}+y_{i},z_{j}+y_{j};z_{-\{i,j\}})\nonumber\\
    =&0,
\end{align*}
where the first equality is obtained using the definition of $h_{ij}$ (see ~\eqref{char-agg-pot-def-1}), while the second
equality is obtained by applying \eqref{char-agg-pot-def-0}. Since $\mathcal{Q}$ is an arbitrary path of length $4$ in the action space, this game is potential by Theorem 2.8 of~\cite{monderer1996potential}.

(b) The result follows from part (a) by replacing each set $K_i$ with $\mathbb{R}^{n_i}.$
\end{proof}
\begin{rem}
    Theorem~\ref{theorem 7} can be viewed as an alternative to Theorem 2.8 (part 4) of \cite{monderer1996potential}.
    Under additional assumptions on the strategy sets, as compared to Theorem 2.8 (part 4) of \cite{monderer1996potential},
    Theorem~\ref{theorem 7} enables us to construct a potential function through the use of Theorem~\ref{nec-suff-cond-2-agg-pot}. While Theorem 2.8 of~\cite{monderer1996potential}
    does not use special assumptions, it does not provide a systematic way
    for constructing a potential function.
\end{rem}

The following example illustrates an application of Theorem~\ref{theorem 7}  to a 3-player Cournot game.

\begin{example} Consider a 3-player Cournot game with 
the following cost functions:
\begin{align}
    &f_1(x_{1},x_{2},x_{3})= (a-b\bar{x})x_1-cx_1, \\
    &f_2(x_{1},x_{2},x_{3})= (a-b\bar{x})x_2-cx_2, \\
    &f_3(x_{1},x_{2},x_{3})= (a-b\bar{x})x_3-cx_3,
\end{align}
where $x_i\in \mathbb{R}$ 
is the decision variable of player $i$, for $i\in\{1,2,3\}$,
and $a,b$ and $c$ are scalars.
We choose players 1 and 2 as an example to verify that 
the game is potential by using Theorem~\ref{theorem 7}(b). For this we need to check that relation~\eqref{char-agg-pot-def-0} holds. We have
\begin{align}\label{EXM-2-1}
    &h_{12}(z_{1},z_{2},y_{1},y_{2};z_3)\nonumber \\
    =&\left(a-b(\bar{z}+y_1)\right)(z_1+y_1)-c(z_1+y_1)\nonumber \\
    &-\left((a-b\bar{z})z_1-cz_1\right)\nonumber \\
    &+\left(a-b(\bar{z}+y_1+y_2)\right)(z_2+y_2)-c(z_2+y_2)\nonumber \\
    &-\left((a-b(\bar{z}+y_1))z_2-cz_2\right).
\end{align}

Moreover, we have
\begin{align}\label{EXM-2-2}
   &h_{12}(0,0,z_{1}+y_{1},z_{2}+y_{2};z_3)-h_{12}(0,0,z_{1},z_{2};z_3)\cr
   = &
   \left(a-b(z_1+z_3+y_1)\right)(z_1+y_1)-c(z_1+y_1)\nonumber \\
   &+\left(a-b(\bar{z}+y_1+y_2)\right)(z_2+y_2)-c(z_2+y_2)\nonumber \\
   &-\left((a-b(z_1+z_3))z_1-cz_1\right)\nonumber \\
   &-\left((a-b\bar{z})z_2-cz_2\right).
\end{align}
We can verify that \eqref{EXM-2-1} and \eqref{EXM-2-2} are equivalent. Thus, the game is potential.
\end{example}

According to Theorem~\ref{theorem 7}, the condition for a game to be potential reduces to checking some equalities in terms of the functions $h_{ij}(\cdot)$ for every $i,j \in \mathcal{N}$. Thus, we may employ Theorem~\ref{theorem 7} to determine a potential function of an $N$-player potential game in terms of these functions. The following theorem provides a characterization of the potential function for  a potential game. 
\begin{thm}\label{theorem 8}
Let an $N$-player game $\Gamma=(\mathcal{N},\{f_i,K_i\}_{i\in\mathcal{N}})$ be potential. Assume that either one of the following cases holds true:
\begin{enumerate}
    \item [(a)] $0\in K$ and the set
$K_i$ is symmetric for all $i \in \mathcal{N}$.
    \item [(b)] The function $f_i(\cdot)$ for $i\in \mathcal{N}$ is defined on $\mathbb{R}^{\bar{n}}$, for all $i\in\mathcal{N}$.
\end{enumerate}

Let $z\in K$ and $z+y\in K$ be arbitrary, and 
consider the path $P$ from $z$ to $z+y$, i.e., 
$P:z\rightarrow(z_1+y_1,z_{-1})\rightarrow(z_1+y_1,z_2+y_2,z_{-\{1,2\}})\rightarrow \ldots\rightarrow z+y$. 
Consider the
function $\phi(\cdot)$ defined by:
if $N=2k+1$ for some $k\in \mathbb{N}$, then
\begin{align}\label{theorem 8-1}
    \phi(z)=\phi(0)&+h_{P_3}(0,(z_1,z_2,z_3))\nonumber \\
    &+\sum_{i=2}^k h_{2i,2i+1}(0,0,z_{2i},z_{2i+1};\hat{z}_{2i-1}),
\end{align}
else if $N=2k$, then
\begin{align}\label{theorem 8-2}
    \phi(z)=\phi(0)&+h_{P_2}(0,(z_1,z_2))\nonumber \\
    &+\sum_{i=2}^k h_{2i-1,2i}(0,0,z_{2i-1},z_{2i};\hat{z}_{2i-2}),
\end{align}
where 
$P_3$ and $P_2$ denote the first 3 and the first 2 steps of the path $P$, respectively, while 
 $\hat{z}_i= (z_1,\ldots,z_i,0,\ldots,0)$ for all $i$.  
Then, $\phi(\cdot)$ is a potential function of the game.
\end{thm}
\begin{proof}
We prove the statement when $N=2k+1$ for some $k\in \mathbb{N}$.
In the path $P$, let $i=N-1$
and $j=N$, i.e., $i=2k$ and $j=2k+1$. By Theorem~\ref{nec-suff-cond-2-agg-pot}, in a potential game, we have that relation~\eqref{theorem-item-3-1} holds with $y=(0,\ldots,z_{N-1},z_{N})$ and $z=(z_1,\ldots,z_{N-2},0,0)$, implying that
\begin{align}\label{theorem 8-3}
& h_{2k,2k+1}\big(0,0,z_{2k},z_{2k+1};\hat{z}_{2k-1}\big)=h_{P}(z,y)\nonumber \\
    &=h_{P}(0,(z_1,\ldots,z_{N-2},z_{N-1},z_{N}))\nonumber \\
    &\ \ -h_{P}(0,(z_1,\ldots,z_{N-2},0,0)).
\end{align}
Additionally, we know that by Theorem~\ref{nec-suff-cond-2-agg-pot} we have
\begin{align}\label{theorem 8-4}
    &\phi(z_1,\ldots,z_{N-2},z_{N-1},z_{N})\nonumber \\
    &=\phi(0)+h_{P}(0,(z_1,\ldots,z_{N-2},z_{N-1},z_{N})).
\end{align}
Rearranging \eqref{theorem 8-3} and substituting in \eqref{theorem 8-4} we have
\begin{align}\label{theorem 8-5}
    &\phi(z_1,\ldots,z_{N-2},z_{N-1},z_{N})\nonumber \\
    &=\phi(0)+h_{P}(0,(z_1,\ldots,z_{N-2},0,0))\nonumber \\
    &\ \ +h_{N-1,N}\big(0,0,z_{N-1},z_{N};\hat{z}_{N-2}\big).
\end{align}
In \eqref{theorem 8-5}, we have that $N-2$ is again an odd number. Moreover,  the $N-2$-player game, obtained by omitting players $N$ and $N-1$ from game $\Gamma$, satisfies  the conditions of Theorem~\ref{theorem 7}. Hence, the $N-2$-player game is potential, and we can repeat the preceding argument. Continuing this process, we eventually reach~\eqref{theorem 8-1}, which completes the proof for the case when the number $N$ of players is odd.
The proof for the other case is identical and it is omitted.
\end{proof}

While Theorem~\ref{theorem 7} provides a necessary and sufficient condition for a game to be potential,
Theorem~\ref{theorem 8} goes further by providing a systematic way to construct a potential function of the game.
The necessary and sufficient condition in terms of function value differences obtained in Theorem~\ref{theorem 7} is comparable with~\cite[Theorem 4.5]{monderer1996potential} for games with one-dimensional action space and differentiable cost functions. 
However, 
an alternative to Theorem~\ref{theorem 8}, which offers a potential function construction for potential games has not been proposed in the existing literature.


\subsection{Examples}\label{sec:example}
In this subsection, we demonstrate the use of Theorems~\ref{theorem 7} and~\ref{theorem 8} in a 
construction of a potential function in potential games through some examples. 

\textbf{Example 2. (Aggregative Game)} Suppose that we have the following utility functions:
\begin{align*}
    &f_1(x_{1},x_{2},x_{3},x_{4})= (a-b\bar{x})x_1-cx_1, \\
    &f_2(x_{1},x_{2},x_{3},x_{4})= (a-b\bar{x})x_2-cx_2, \\
    &f_3(x_{1},x_{2},x_{3},x_{4})= (a-b\bar{x})x_3-cx_3,\\
    &f_4(x_{1},x_{2},x_{3},x_{4})= (a-b\bar{x})x_4-cx_4,
\end{align*}
where $x_i\in \mathbb{R}$ is the decision variable of player $i$ for $i\in\{1,2,3,4\}$.
Using Theorem~\ref{theorem 8}, we have
\begin{align}\label{example-2}
\phi(x)=\phi(0)&+(a-bx_1)x_1-cx_1\nonumber \\
&+(a-b(x_1+x_2))x_2-cx_2\nonumber \\
&+(a-b(x_1+x_2+x_3))x_3-cx_3\nonumber \\
&+(a-b(x_1+x_2+x_3+x_4))x_4-cx_4.
\end{align}
We have that $\phi(x_1+y_1,x_2,x_3,x_4)-\phi(x_1,x_2,x_3,x_4)=(a-b \bar{x})y_1-bx_1y_1-by_1^2-cy_1$, which can be seen to be equal to the difference  $f_1(x_1+y_1,x_2,x_3,x_4)-f_1(x_1,x_2,x_3,x_4)$.

In aggregative games, instead of $z_{-\{i,j\}}$, the term $\bar{z}_{-\{i,j\}}$ appears in the derivations. Therefore, the conditions in Theorem~\ref{theorem 7}, in particular \eqref{char-agg-pot-def-0}, need only to be satisfied for every fixed strategy $\bar{z}_{-\{i,j\}} \in \bar{K}_{-\{i,j\}}$. 

\textbf{Example 3. (Network Congestion Game)} 
Here, we consider the game of network congestion, which was introduced by Rosenthal in 1973 \cite{C23}. Different aspects of this problem have been examined extensively in the literature, including the complexity of obtaining a solution for maximizing social welfare \cite{C24}, analyzing the price of anarchy and stability \cite{Christodoulou2}, determining whether there is a Nash equilibrium \cite{Facchini,Ui,monderer1996potential}, and checking whether it is a potential game \cite{monderer1996potential}, \cite{C25}. 
This type of games is proven to be potential by
Rosenthal \cite{C23}, who has provided a potential function, and also later by Monderer and Shapely \cite{monderer1996potential}. 
However, no systematic way to obtain the potential function is introduced in the literature. Our Theorems~\ref{theorem 7} and \ref{theorem 8} provide a systematic approach to constructing a potential function for this game. 

For simplicity, we focus on the variant of the problem where $N$ players are departing from a point $O$ in a given network $G(V,E)$, with node set $V$ and edge set $E$, to reach their destination at point $D$, where $O,D\in V$. In this problem, the cost of using link $e\in E$ is equivalent to $C_e(v_e)$, where $v_e$ is the total number of players who choose link $e$ in their path from $O$ to $D$, and $C_e:\mathbb{N}\rightarrow \mathbb{R}$ is the cost of using link $e$, which is a function of the total number of players who choose this link. For the sake of simplicity, let us assume there are $m\in \mathbb{N}$ different paths from the origin $O$ to the destination $D$, and there is no restriction on choosing among them for any player. Therefore, the action space for each player $i$ is $K_i:=\{p_1,p_2,\ldots,p_m\}$ where $p_l$ is the $l$th path available to each player. Hence, the cost function for player $i\in \mathcal{N}=\{1,2,\ldots,N\}$ is \[f_i(x_1,\ldots,x_N)=\sum_{e\in x_i} C_e\big(v_e(x_1,\ldots,x_N)\big),\] 
where the summation is over the links $e$ traversed in a path $x_i$.
Every player aims to minimize his cost function \cite{C23}. 

To apply Theorem~\ref{theorem 7} or Theorem~~\ref{theorem 8}, we need 
to transform the game to an appropriate form. 
Let us assign an arbitrary order to the elements of the action space, i.e., the path collection $\{p_1,\ldots,p_m\}$, so that each player $i$ action set is $A_i:=\{1,2,\ldots,m\}$, as an alternative to set $K_i$. 
Moreover, let us consider another network $G(V',E')$ 
such that $V\cap V'=\{O\}$, $E\cap E'=\emptyset$, and there is a destination $D'\in V'$ that can be reached from $O$ via $m$-different paths.
Additionally, we add a self-loop to node $O$, which is denoted as path $p_0$. 
We now consider a game where each player $i$ has 
action space $A_i:=\{-m,\ldots,-1,0,1,2,\ldots,m\}$, where $-j\in A_i$ denotes the path $p_j$ connecting $O$ to $D'\in G(V',E')$. 
Assume
that $C_e(v_e)$ is very large for all $e\in E'$ in comparison to all $e\in E$ for every possible value of $v_e$. We also choose the cost of using the self-loop of $O$ large enough so that no one has an incentive to choose it. 
 As a result, in this newly constructed game, players can choose a path to reach $D$ on the actual network, $D'$ on the artificial network, or stay at the origin $O$. Clearly, with the costs as described above, nobody has the incentive to move into the artificial network $G(V',E')$. Therefore, this game on an augmented network is identical to the network congestion game on $G(V,E)$. \\ 
We will show that the game on the augmented network is potential and construct a potential function by using Theorem~\ref{theorem 7}, i.e., by using~\eqref{char-agg-pot-def-0}. Let us assume all players apart from $i,j$ have already taken their decision, and $v_e$ is the number of users for any link $e$, excluding players $i,j$. 
Assume player $i$ chooses path $l_i$ and later deviates to path $l_i'$ and, similarly, player $j$ selects $l_j$ and later on deviates to $l_j'$. 
For simplicity, we consider the case when 
there is no link in common among these paths; however, similar arguments can be presented for the case when there are links in common. Therefore, the left-hand side of \eqref{char-agg-pot-def-0} is
\begin{equation} \label{exmp3-1}
 \sum_{e\in l_i'}C_e(v_e+1)-\sum_{e\in l_i}C_e(v_e+1)+\sum_{e\in l_j'}C_e(v_e+1)-\sum_{e\in l_j}C_e(v_e+1).
\end{equation}
When players $i$ and $j$, rather than remaining at the origin, choose to deviate to $l_i'$ and $l_j'$, respectively,
results in a difference in the cost of
\begin{equation}\label{exmp3-2}
 \sum_{e\in l_i'}C_e(v_e+1)+\sum_{e\in l_j'}C_e(v_e+1)-\sum_{e\in l_0}C_e(v_e)-\sum_{e\in l_0}C_k(v_e),
\end{equation}
and similarly, rather than remaining at the origin, deviating to $l_i$ and $l_j$ for players $i,j$ causes a difference in the cost of
\begin{equation}\label{exmp3-3}
     \sum_{e\in l_i}C_e(v_e+1)+\sum_{e\in l_j}C_e(v_e+1)-\sum_{e\in l_0}C_e(v_e)-\sum_{e\in l_0}C_e(v_e).
\end{equation}
Subtracting \eqref{exmp3-3} from \eqref{exmp3-2} results in \eqref{exmp3-1}. This shows by Theorem~\ref{theorem 7} that the game is potential. Note that by checking \eqref{shap_c} for an arbitrary cycle of length $4$, we will see that the right-hand side of \eqref{shap_c} is zero - which alternatively shows that the game is potential. 

To obtain the potential function using Theorem~\ref{theorem 8}, we note that this theorem enables us to load the network incrementally. As immediate result of applying Theorem \ref{theorem 8}, where at each step two players are selected and assigned to their route in the network, 
the following potential function is obtained
\[\phi(x)= \phi(0)-\sum_{k=1}^{N}C_0(k)+\sum_{e\in E}\sum_{k=1}^{v_e(x)}C_e(k),\]
where $\phi(0)$ is some constant, $C_0$ is the cost of using the self loop at the origin, and $v_e(x)$ is the number of users using link $e$ when decision variables are set to $(x_1,\ldots,x_N)$. The non-constant part of this potential function is identical to that which was introduced originally in~\cite{C23}. 
\section{Some Characterization of Ordinal and Generalized Ordinal Potential Games}\label{Main-res-2}
In this section we provide some sufficient conditions for a game to be an ordinal potential game and a generalized ordinal potential game. 
We start by considering  the case of an ordinal potential game by using the following assumption.
\begin{assum}\label{generic-suff-condition}
Consider a game $\Gamma=(\mathcal{N},\{f_i,K_i\}_{i\in\mathcal{N}})$. Assume that for all $i,j\in \mathcal{N}$, $x_i,x_i+y_i \in K_i$, $x_j,x_j+y_j \in K_j$, and $x_{-\{i,j\}}\in K_{-\{i,j\}}$, we have
    \begin{align}\label{fund-condin}
        &f_i(x_i+y_i,x_j+y_j,x_{-\{i,j\}})-f_i(x_i,x_{-i})<0 \iff \nonumber\\
        &f_j(x_j+y_j,x_i+y_i,x_{-\{i,j\}})-f_j(x_j,x_{-j})<0.
    \end{align}
\end{assum}
Under Assumption~\ref{generic-suff-condition}, we have the following result.
\begin{thm}\label{generic-suff-theo}
     Under Assumption~\ref{generic-suff-condition}, 
     a game $\Gamma=(\mathcal{N},\{f_i,K_i\}_{i\in\mathcal{N}})$ is ordinal potential with $f_i(\cdot)$ as ordinal potential function for any $i\in \mathcal{N}$.
\end{thm}
\begin{proof}
    Since~\eqref{fund-condin} holds for any $y_i$ such that $x_i+y_i \in K_i$, we can consider $y_i=0$. Therefore, for any $j\in \mathcal{N}$ we have
     \begin{align*}
        &f_i(x_i,x_j+y_j,x_{-\{i,j\}})-f_i(x_i,x_{-i})<0 \iff \nonumber\\
        &f_j(x_j+y_j,x_i,x_{-\{i,j\}})-f_j(x_j,x_{-j})<0.
    \end{align*}
    This means by Definition~\ref{def_Best_potential} the game $\Gamma=(\mathcal{N},\{f_i,K_i\}_{i\in\mathcal{N}})$ is ordinal potential with ordinal potential function $f_i(\cdot)$ for $i\in \mathcal{N}$.
\end{proof}

When the player's cost functions are twice continuously differentiable and the action space of each player is one dimensional, Assumption~\ref{generic-suff-condition} is equivalent to 
    \begin{equation} \label{local-p}
        \nabla^2_{x_i,x_j} f_i(x)<0 \iff \nabla^2_{x_j,x_i}f_j(x)<0\qquad\hbox{for all }x\in K,
    \end{equation}

In~\cite{ewerhart2020ordinal} 
an analogous local version of~\eqref{local-p} has been shown as a necessary condition for a game to be generalized ordinal potential. In particular, for a game with smooth costs to be generalized ordinal potential, it is necessary that 
following relation holds
\begin{equation} \label{local-p-NE}
    \nabla^2_{x_i,x_j} f_i(x^*)\cdot \nabla^2_{x_j,x_i}f_j(x^*)\geq 0,
    \end{equation}
for all points $x^*$ in the joint action space satisfying $\nabla_{x_i} f_i(x^*)= \nabla_{x_j} f_j(x^*)=0$ for all $i,j$. 
Our Theorem~\ref{generic-suff-theo} provides a sufficient condition for such a game to be ordinal potential.
Specifically, by Theorem~\ref{generic-suff-theo}, we have that if~\eqref{local-p} holds globally, then the game is ordinal potential.

Next, we derive a sufficient condition for existence of a generalized ordinal potential for a game with strongly convex cost functions with Lipschitz continous gradients.
Such functions can be 
characterized by the first-order condition, as follows.
Let $X\subseteq\mathbb{R}^n$ be a nonempty convex set and
$f:X\to \mathbb{R}$ be a continuously differentiable function on $X$. Then, $f(\cdot)$ 
is strongly convex on $X$ with a constant $\eta>0$ if and only if 
the following relation holds  in~\cite{boyd2004convex},\cite{Bertsekas}: 
for all $x,y\in X$,
\begin{equation}\label{strict-conv-func}
f(y) \geq f(x)+\langle \nabla f(x),y-x \rangle+\frac{\eta}{2}\|y-x\|^2.
\end{equation}
A continuously differentiable function
$f(\cdot)$ 
is strictly convex on a convex set $X$ if and only if
the following relation holds~\cite{boyd2004convex},\cite{Bertsekas}: 
for all $x,y\in X$, with $x\ne y$, 
\begin{equation}\label{strong-conv-func}
f(y) > f(x)+\langle \nabla f(x),y-x \rangle.
\end{equation}

A function $f(\cdot)$ has Lipschitz continuous gradients on a set $X$ with a constant $L>0$ when the following relation holds:  $$\|\nabla f(x)-\nabla f(y)\|\leq L\|x-y\|
\qquad\hbox{for all $x,y \in X $}.
$$
When a function $f(\cdot)$ has Lipschitz continuous gradients on a convex set $X$, with a constant $L>0$, then  the following inequality is valid
{\cite{nesterov2018lectures},\cite{beck2014introduction}}: for all $x,y\in X$,
\begin{align}\label{Grad-Lip-func}
f(y) \leq f(x)+\langle \nabla f(x),y-x \rangle+\frac{L}{2}\|y-x\|^2.
\end{align}

\begin{thm}\label{final-the}
    Consider a game $\Gamma=(\mathcal{N},\{f_i,K_i\}_{i\in\mathcal{N}})$
     where each strategy set $K_i$ is convex and
     each cost function $f_i(\cdot,x_{-i})$ is strongly convex for all $x_{-i}\in K_i$ with a constant $\eta_i$. 
     Assume there exists a 
     differentiable function $\phi(\cdot)$ such that:
     \begin{itemize}
        \item[(a)] For every $i\in\mathcal{N}$, $x_i,y_i\in K_i$, and $x_{-i}\in K_{-i}$, if $\la \nabla_{x_i} f_i(x_i,x_{-i}),y_i-x_i\ra <0$, then
        \[\la \nabla_{x_i} \phi(x_i,x_{-i}),y_i-x_i\ra
        \le \la \nabla_{x_i} f_i(x_i,x_{-i}),y_i-x_i\ra.\]
        \item[(b)] The function $\phi(\cdot)$
        has a Lipschitz continuous gradients with a constant  $L>0$ such that
    \[L\le \min_{1\le i\le N} \eta_{i}.\]
     \end{itemize}
 Then, the game is a generalized ordinal potential game with a generalized ordinal potential function $\phi(\cdot)$.
\end{thm}

\begin{proof}
Let $i\in \mathcal{N}$ be an arbitrary player. Consider
    $x_i,y_i\in K_i$ with $x_i\ne y_i$ and
    $x_{-i}\in K_{-i}$, and assume that 
    \[f_i(y_i,x_{-i}) - f_i(x_i,x_{-i})<0.\]
    By the strong convexity of $f_i(\cdot)$, we have for 
    $x=(x_i,x_{-i})\in K$ and $y=(y_i,x_{-i})\in K$,
    \[f_i(x)+\la\nabla_{x_i} f(x),y_i-x_i\ra +\frac{\eta_i}{2}\|y_i-x_i\|^2\le f_i(y_i,x_{-i}).\]
    By combining the preceding two relations, we find that
\begin{equation}
    \label{eq-less1}
\la\nabla_{x_i} f(x),y_i-x_i\ra +\frac{\eta_i}{2}\|y_i-x_i\|^2 < 0.
\end{equation}
In view of the preceding relation, it follows that 
\begin{equation}
    \label{eq-less0}
    \la\nabla_{x_i} f(x),y_i-x_i\ra<0.
    \end{equation}

On the other hand by our assumption the function $\phi(\cdot)$ has Lipschitz continuous gradients and, thus, satisfies for $x=(x_i,x_{-i})$
and $y=(y_i,x_{-i})$,
    \begin{align*}
        \phi(y)
        &\leq \phi(x)+\langle \nabla \phi(x),y-x\rangle +\frac{L}{2}\|y-x\|^2 \cr
        &=\phi(x)+\langle \nabla_{x_i}\phi(x),y_i-x_i\rangle +\frac{L}{2}\|y_i-x_i\|^2.
        \end{align*}
By relation~\eqref{eq-less0} and the assumed property (a) for the function $\phi(\cdot)$, it follows that
\begin{align*} 
        \phi(y)
&\le \phi(x)+\langle\nabla_{x_i} f_i(x),y_i-x_i\rangle+\frac{L}{2}\|y_i-x_i\|^2\cr
&\le \phi(x)+\langle\nabla_{x_i} f_i(x),y_i-x_i\rangle+\frac{\eta_i}{2}\|y_i-x_i\|^2\cr
&<\phi(x),
    \end{align*}
    where the second inequality follows by $L\le \eta_i$ for all $i$ (see the assumed property (b)), while the last inequality follows from~\eqref{eq-less1}.
    Therefore, $\phi(\cdot)$  is a generalized ordinal potential function for the game $\Gamma$.
\end{proof}
\def\a{\alpha}

We next consider a variant of
Theorem~\ref{final-the}, where we relax the strong convexity of the cost functions. In this case, we will require the existence of a function $\phi(\cdot)$ that is concave over the set $K$, i.e.,
for all $x,y\in K$, and $\a\in[0,1]$,
\[\phi(\a x+(1-\a)y)\ge \a \phi(x)+(1-\a)\phi(y).\]

In the sequel, we will need a concept of a subgradient of a concave function $h(\cdot).$
We say that a vector $s$ is a subgradient of a concave function 
$h(\cdot)$ at a point $x$, if the following relation holds:
\[h(y)\le h(x) +\la s,y-x\ra \quad\hbox{for all $y\in{\rm dom}(h)$},\]
where ${\rm dom}(h)$ denotes the domain of the function $h(\cdot)$.

For a game with a strictly convex functions in players' decision variables, we have the following result.
\begin{thm}\label{final-the-one}
    Consider a game $\Gamma=(\mathcal{N},\{f_i,K_i\}_{i\in\mathcal{N}})$
     where each strategy set $K_i$ is convex and
     each cost function $f_i(\cdot,x_{-i})$ is strictly convex over the set $K_i$ for every $x_{-i}\in K_{-i}$. 
     Assume that there exists a 
     function $\phi(\cdot)$ on the set $K$ that is concave in $x_i\in K_i$ and has a subgradient at each $x_i\in K_i$, for every $x_{-i}\in K_{-i}$. 
     Moreover, 
     for every $i\in\mathcal{N}$, $x_i,y_i\in K_i$, and $x_{-i}\in K_{-i}$, if $\la \nabla_{x_i} f_i(x_i,x_{-i}),y_i-x_i\ra <0$, there exists a subgradient $s_i(x_i,x_{-i})$ of $\phi(\cdot,x_{-i})$ at $x_i\in K_i$  satisfying
        \[\la s_i(x_i,x_{-i}),y_i-x_i\ra
        \le \la \nabla_{x_i} f_i(x_i,x_{-i}),y_i-x_i\ra.\]
        Then, $\phi(\cdot)$ is a generalized ordinal potential for the game.
     \end{thm}
 \begin{proof}
Let $i\in \mathcal{N}$ be an arbitrary player. Consider
    $x_i,y_i\in K_i$ with $x_i\ne y_i$ and
    $x_{-i}\in K_{-i}$, and let
    \[f_i(y_i,x_{-i}) - f_i(x_i,x_{-i})<0.\]
    By the strict convexity of $f_i(\cdot)$, we have for 
    $x=(x_i,x_{-i})\in K$ and $y=(y_i,x_{-i})\in K$,
    \[f_i(x)+\la\nabla_{x_i} f(x),y_i-x_i\ra < f_i(y_i,x_{-i}).\]
    The preceding two relations yield
    \begin{equation}
    \label{eq-less11}
\la\nabla_{x_i} f(x),y_i-x_i\ra < 0.
\end{equation}

   By our assumption the function $\phi(\cdot)$ is concave and has  subgradients on the set $K$, so it satisfies the following relation: for $x=(x_i,x_{-i})\in K$, and $y=(y_i,x_{-i})\in K$, and some subgradient $s_i(x)$ of $\phi(\cdot,x_{-i})$ at the point $x$, so that
    \begin{align*}
        \phi(y)
        &\leq \phi(x) +\langle s_i(x),y_i-x_i\rangle.
        \end{align*}
By relation~\eqref{eq-less11} and the assumed property for the function $\phi(\cdot,x_{-i})$, it follows that
\begin{align*} 
        \phi(y)
&\le \phi(x)+\langle\nabla_{x_i} f_i(x),y_i-x_i\rangle\cr
&<\phi(x).
 \end{align*}
    Hence, $\phi(\cdot)$  is an ordinal potential function for the game $\Gamma$.
\end{proof}
    
The next example is designed to show that the sufficient condition obtained in Theorem~\ref{final-the-one} can capture nontrivial generalized ordinal potential games.

\textbf{Example 4.}  Consider following cost functions
\begin{align*}
    &f_1(x_1,x_2)=(x_1+x_2)^2 \qquad\hbox{for $x_1,x_2\in(0,1]$},\cr
    &f_2(x_1,x_2)=(x_1+x_2)^6\qquad \hbox{for $x_1,x_2\in(0,1]$}.
\end{align*}
By Theorem~\ref{generic-suff-theo}, we can verify that  the game is ordinal potential, where both $\phi_1(x_1,x_2)=(x_1+x_2)^6$ and $\phi_2(x_1,x_2)=(x_1+x_2)^2$ are
ordinal potential functions.
However, we will construct a generalized ordinal potential function $\phi(\cdot)$ 
based on Theorem~\ref{final-the-one}.
First, 
we note that the sets $K_i=(0,1]$, $i=1,2,$ are convex and each cost function $f_i(\cdot, x_{-i})$ is strictly convex on $K_i$. Next, for the partial derivatives $\nabla_if_i(x_1,x_2)$ we have
\[\nabla_{x_1} f_1(x_1,x_2)=2(x_1+x_2)>0
\qquad \hbox{for all }x_1,x_2\in(0,1],\]
\[\nabla_{x_2} f_2(x_1,x_2)=6(x_1+x_2)^5>0
\qquad \hbox{for all }x_1,x_2\in(0,1].\]
Now, consider the function $\phi(\cdot)$ of the form:
\[\phi(x_1,x_2)=a\sqrt{x_1}+ b\sqrt{x_2},\ 
\hbox{for all }x_1,x_2\in(0,1],\]
where $a,b>0$. This function is concave, so it is concave in each of the variables individually.
Next, we choose the coefficients $a,b>0$ so that the condition on the partial gradients in Theorem~\ref{final-the-one} is satisfied. Since 
\[\nabla_{x_1}\phi(x_1,x_2)=\frac{a}{2\sqrt{x_1}}, \qquad
\nabla_{x_2}\phi(x_1,x_2)= \frac{b}{2\sqrt{x_2}},\]
we have that 
\[\nabla_{x_1}\phi(x_1,x_2)\ge \frac{a}{2}, \quad
\nabla_{x_2}\phi(x_1,x_2)\ge \frac{b}{2},
\quad
\hbox{for }x_1,x_2\in(0,1].\]
By choosing $a\ge 8$ and $b\ge 6\cdot 2^6$,
we can see that for all $x_1,x_2\in(0,1]$,
\[\nabla_{x_1}\phi(x_1,x_2)\ge 4\ge \nabla_{x_1}f_1(x_1,x_2),\]
\[\nabla_{x_2}\phi(x_1,x_2)\ge 6\cdot 2^5
\ge \nabla_{x_1}f_1(x_1,x_2).\]
To have $\la \nabla_{x_i} f_i(x_i,x_{-i}),y_i-x_i\ra <0$ for $i\in \{1,2\}$, the term $y_i-x_i$ must be negative. As a result, we can see that for all $x_1,x_2\in(0,1]$,
\[\la\nabla_{x_1}\phi(x_1,x_2),y_1-x_1\ra\le \la\nabla_{x_1}f_1(x_1,x_2),y_1-x_1\ra,\]
\[\la\nabla_{x_2}\phi(x_1,x_2),y_2-x_2\ra\le \la\nabla_{x_1}f_1(x_1,x_2),y_2-x_2\ra,\]
implying that the condition on the partial gradients in Theorem~\ref{final-the-one} is satisfied. Thus, $\phi(\cdot)$ is a generalized ordinal potential for the game.

For a game with a strictly convex functions in players' decision variables, we also have the following result.
\begin{thm}\label{final-the-two}
    Consider a game $\Gamma=(\mathcal{N},\{f_i,K_i\}_{i\in\mathcal{N}})$
     where each strategy set $K_i$ is convex and
     each cost function $f_i(\cdot,x_{-i})$ is strictly convex over the set $K_i$ for every $x_{-i}\in K_{-i}$. 
     Assume that there exists a 
     function $\phi(\cdot)$ on the set $K$ that is concave in $x_i\in K_i$ and has a subgradient at each $x_i\in K_i$, for every $x_{-i}\in K_{-i}$. 
     Moreover, 
     for every $i\in\mathcal{N}$, $x_i,y_i\in K_i$, and $x_{-i}\in K_{-i}$, if $\la \nabla_{x_i} f_i(x_i,x_{-i}),y_i-x_i\ra <0$, there exists a subgradient $s_i(x_i,x_{-i})$ of $\phi(\cdot,x_{-i})$ at $x_i\in K_i$ and some scalar function $\alpha_i(x)$ with $\alpha_i(x)>0$ for every $x\in K$  satisfying
        \[\la s_i(x_i,x_{-i}),y_i-x_i\ra
        \le \la \alpha_i(x)\nabla_{x_i} f_i(x_i,x_{-i}),y_i-x_i\ra.\]
        Then, $\phi(\cdot)$ is a generalized ordinal potential for the game.
     \end{thm}
 \begin{proof}
Let $i\in \mathcal{N}$ be an arbitrary player. Consider
    $x_i,y_i\in K_i$ with $x_i\ne y_i$ and
    $x_{-i}\in K_{-i}$, and let
    \[f_i(y_i,x_{-i}) - f_i(x_i,x_{-i})<0.\]
    By the strict convexity of $f_i(\cdot)$, we have for 
    $x=(x_i,x_{-i})\in K$ and $y=(y_i,x_{-i})\in K$,
    \[f_i(x)+\la\nabla_{x_i} f(x),y_i-x_i\ra < f_i(y_i,x_{-i}).\]
    The preceding two relations yield
    \begin{equation}
    \label{eq-less111}
\la\nabla_{x_i} f(x),y_i-x_i\ra < 0.
\end{equation}

   By our assumption the function $\phi(\cdot)$ is concave and has  subgradients on the set $K$, so it satisfies the following relation: for $x=(x_i,x_{-i})\in K$, and $y=(y_i,x_{-i})\in K$, and some subgradient $s_i(x)$ of $\phi(\cdot,x_{-i})$ at the point $x$, so that
    \begin{align*}
        \phi(y)
        &\leq \phi(x) +\langle s_i(x),y_i-x_i\rangle.
        \end{align*}
By relation~\eqref{eq-less111} and the assumed property for the function $\phi(\cdot,x_{-i})$, it follows that
\begin{align*} 
        \phi(y)
&\le \phi(x)+\langle\alpha_i(x)\nabla_{x_i} f_i(x),y_i-x_i\rangle\cr
&<\phi(x).
 \end{align*}
    Hence, $\phi(\cdot)$  is an ordinal potential function for the game $\Gamma$.
\end{proof}
    
The next example is designed to show that the sufficient condition obtained in Theorem~\ref{final-the-two} can capture nontrivial generalized ordinal potential games.

\textbf{Example 5.}  Consider following cost functions
\begin{align*}
    &f_1(x_1,x_2)=(x_1+x_2)^2 \qquad\hbox{for $x_1,x_2\in(0,\infty]$},\cr
    &f_2(x_1,x_2)=(x_1+x_2)^6\qquad \hbox{for $x_1,x_2\in(0,\infty]$}.
\end{align*}
We will construct a generalized ordinal potential function $\phi(\cdot)$ 
based on Theorem~\ref{final-the-two}. We can simply check that a generalized ordinal potential function $\phi(x)=2(x_1+x_2)^{0.4}$ along with $\alpha_1(x)=\frac{4}{10(x_1+x_2)^{1.6}}$ and $\alpha_2(x)=\frac{4}{30(x_1+x_2)^{5.6}}$ satisfy the conditions in Theorem~\ref{final-the-two}. Thus, $\phi(\cdot)$ is a generalized ordinal potential for the game.
\section{Conclusions}\label{sec:conclusion}
In this paper, we first derived a necessary and sufficient condition for games to fall under the category of potential games. These conditions do not
require any differentiation or integration processes in constrast to
those obtained in \cite{Hwang}. We stepped further and simplified the general criteria we obtained for potential games for the class of aggregative games. This relation completely describes aggregative potential games in terms of every two players' cost functions and coupling behavior. We checked the condition through a $3$-player Cournot game. 

Under the assumptions of action space containing zero or the cost functions being
definable on the entire $\mathbb{R}^{\bar{n}}$, we can find a useful closed form expression for the potential function in terms of those functions constructed for the evaluation to be potential game. We also examined the form of potential function for potential games through an example of $4$-player aggregative game, as well as an example of a network congestion game. 

Moreover, we proposed some characterizations of ordinal potential games and generalized ordinal potential games. In comparison to the results of \cite{ewerhart2020ordinal} on local necessary conditions for smooth games to be ordinal potential, we obtained analogous global sufficient condition for games with possibly non smooth cost functions to be ordinal potential.  Additionally, we proposed some characterization of generalized ordinal potential games where cost functions are strongly/strictly convex in terms of their player decision variable. We also provide an example for a $2$-player ordinal potential game.

\bibliographystyle{plain}        
\bibliography{main} 
\end{document}